\documentclass[11pt]{article}
\usepackage{amssymb}
\usepackage{amsmath}
\usepackage{amsthm}
\usepackage{latexsym}
\usepackage{graphicx}
\newtheorem{theo}{Theorem}[section]
\newtheorem{prop}[theo]{Proposition}
\newtheorem{cor}[theo]{Corollary}
\newtheorem{lemma}[theo]{Lemma}
\theoremstyle{definition}
\newtheorem{defi}[theo]{Definition}
\newtheorem{exa}[theo]{Example}
\newtheorem{rem}[theo]{Remark}
\numberwithin{equation}{section}

\newcommand{\DS}{\displaystyle}

\newcommand{\F}{{\mathbb F}}

\newcommand{\Z}{{\mathbb Z}}

\newcommand{\C}{{\mathbb C}}
\newcommand{\R}{{\mathbb R}}

\newcommand{\cH}{{\mathcal H}}

\newcommand{\cP}{{\mathcal P}}

\newcommand{\Rhat}{\widehat{R}}
\newcommand{\wcP}{\widehat{\mathcal P}}

\newcommand{\cPhom}{{\mathcal P}_{\rm{hom}}}
\newcommand{\cPrank}{{\mathcal P}_{\rm{rk}}}
\newcommand{\wcPhom}{\widehat{{\mathcal P}_{\rm{hom}}}}
\newcommand{\wcPchil}{\wcP^{^{\scriptscriptstyle[\chi,l]}}}
\newcommand{\wcPchir}{\wcP^{^{\scriptscriptstyle[\chi,r]}}}
\newcommand{\wcPhoml}{\widehat{{\mathcal P}_{\rm{hom}}}^{\scriptscriptstyle l}}
\newcommand{\wcPhomr}{\widehat{{\mathcal P}_{\rm{hom}}}^{\scriptscriptstyle r}}

\newcommand{\cQ}{{\mathcal Q}}

\newcommand{\GL}{\mbox{\rm GL}}

\newcommand{\widesim}[1][1.5]{\scalebox{#1}[1]{$\sim$}}

\newcommand{\rank}{\mbox{\rm rk}\,}

\renewcommand{\mod}{\mbox{\rm mod}\,}
\newcommand{\soc}{\mbox{\rm soc}}
\newcommand{\rad}{\mbox{\rm rad}}

\newcommand{\rowspace}{\mbox{\rm rowspace}}
\newcommand{\wt}{{\rm wt}}

\newcommand{\tr}{{\rm tr}}

\newcommand{\mmid}{\mbox{$\,|\,$}}
\newcommand{\mmidbig}{\mbox{$\,\big|\,$}}

\newcommand{\Gaussian}[2]{\genfrac{[}{]}{0pt}{}{#1}{#2}}
\newcommand{\Binom}[2]{\genfrac{(}{)}{0pt}{}{#1}{#2}}

\newenvironment{liste}{\begin{list}{--\hfill}{\topsep0ex \labelwidth.4cm
   \leftmargin.5cm \labelsep.1cm \rightmargin0cm \parsep0ex \itemsep.6ex
   \partopsep1.4ex}}{\end{list}}
\newcounter{alp}
\newcounter{ara}
\newcounter{rom}
\newenvironment{romanlist}{\begin{list}{(\roman{rom})\hfill}{\usecounter{rom}
     \topsep0ex \labelwidth.7cm \leftmargin.7cm \labelsep0cm
     \rightmargin0cm \parsep0ex \itemsep.4ex
     \partopsep1ex}}{\end{list}}
\newenvironment{alphalist}{\begin{list}{(\alph{alp})\hfill}{\usecounter{alp}
     \topsep.5ex \labelwidth.6cm \leftmargin.6cm \labelsep0cm
     \rightmargin0cm \parsep0ex \itemsep0ex
     \partopsep1.6ex}}{\end{list}}
\newenvironment{arabiclist}{\begin{list}{(\arabic{ara})\hfill}{\usecounter{ara}
     \topsep0ex \labelwidth.6cm \leftmargin.6cm \labelsep0cm
     \rightmargin0cm \parsep0ex \itemsep0ex
     \partopsep1.6ex}}{\end{list}}

\title{The Homogeneous Weight Partition and\\ its Character-Theoretic Dual}
\date\today
\author{Heide Gluesing-Luerssen\thanks{The author was partially supported by the National Science Foundation
        grant \#DMS-1210061.}\\
        University of Kentucky\\ Department of Mathematics\\
       715 Patterson Office Tower\\ Lexington, KY 40506-0027, USA\\ heide.gl@uky.edu}

\begin{document}
\maketitle
\noindent{\bf Abstract:}
The values of the normalized homogeneous weight are determined for arbitrary finite Frobenius rings and expressed in
a form that is independent from a generating character and the M\"obius function on the ring.
The weight naturally induces a partition of the ring, which is invariant under
left or right multiplication by units.
It is shown that the character-theoretic left-sided dual
of this partition coincides with the right-sided dual, and even more, the left- and right-sided Krawtchouk
coefficients coincide.
An example is provided showing that this is not the case for general invariant partitions if the ring is
not semisimple.

\medskip
\noindent{\bf Keywords:} Homogeneous weight, finite Frobenius rings, character-theoretic dual partitions

\smallskip
\noindent{\bf MSC (2000):} 94B05, 94B99, 16L60

\section{Introduction}\label{SS-Intro}

In this paper we will study the homogeneous weight on arbitrary finite Frobenius rings.
These weights have been introduced by Constantinescu and Heise in~\cite{CoHe97} and have received a lot of
attention in the ring-linear coding literature ever since.
We refer to the introductions of any of the papers \cite{BGO07,BKS12,GL13homog,GrSch00,Hon01,HoLa01} for
motivation and background on the homogeneous weight.

In this paper we will first present, for an arbitrary finite Frobenius ring, an expression for the values of the
(normalized) homogeneous weight that is different from those based on characters or the M\"obius function.
This extends the results in~\cite{GL13homog} where we restricted ourselves to rings that are direct products of local
Frobenius rings.
We will see that the homogeneous weight has value one exactly for the elements outside the socle.
Moreover, the weight on the socle of~$R$ is closely related to the weight on $R/\rad(R)$, where $\rad(R)$ is the Jacobson radical.
The Wedderburn-Artin decomposition for $R/\rad(R)$ will then reduce the considerations
to studying the homogeneous weight on matrix rings over finite fields.
In that case it is easy to compute the values based on the rank.
It will show that matrices share the same weight if and only if they have the same rank.

After having determined the values of the homogeneous weight in the described form, we will turn to the partition of the ring induced by
this weight and study its dual with respect to character-theoretic dualization.
This dualization plays a central role in the area of MacWilliams identities.
Indeed, if a partition of~$R$ (or $R^n$) is reflexive (i.e., coincides with its character-theoretic bidual), then the
partition enumerator of a code and the dual-partition enumerator of the dual code uniquely determine each other.
This has been discussed in various forms in many papers, see for instance
\cite{BGO07,Cam98,Del73,HoLa01,ZiEr96,ZiEr09} as well as~\cite{GL13pos}, where a general approach in the terminology of
this paper was presented.

For non-commutative rings the above notions all come in a left and right version.
Furthermore, with the left and right dual partitions are associated left and right Krawtchouk coefficients, which determine the actual
MacWilliams transformation between the partition enumerator of a code and that of its (left or right) dual code.

We will show that for the homogeneous weight partition, the left and right dual coincide, and so do the left and right
Krawtchouk coefficients.
As we will see, this is true for any invariant partition (i.e., all partition sets are invariant under left or right
multiplication by units) as long as the ring is semisimple.
But for non-semisimple rings, the particularly close relationship between the homogeneous weight on the socle and the
homogeneous weight on $R/\rad(R)$ will be crucial for the independence of dualization and Krawtchouk coefficients
from the sidedness.

\section{Preliminaries}
We begin with briefly recalling some properties of finite Frobenius rings.
Then we move on to the homogeneous weight and present some basic facts.

Let~$R$ be a finite ring with identity, and let~$R^*$ be its group of units.
Denote by~$\soc(_R R)$ the socle of the left~$R$-module~$R$, and
let~$\rad(R)$ be the Jacobson radical of~$R$.
Moreover, denote by~$\Rhat:=\text{Hom}(R,\C^*)$ the group of characters of~$R$
(i.e., group homomorphisms from $(R,\,+\,)$ to $\C^*$).
Then~$\Rhat$ is an $R$-$R$-bimodule via the left and right scalar multiplications
\begin{equation}\label{e-bimodule}
    (r\!\cdot\!\chi)(v)=\chi(vr)\ \text{ and }\
    (\chi\!\cdot\! r)(v)=\chi(rv)\text{ for all } r\in R\text{ and }v\in R.
\end{equation}
We summarize the crucial properties of finite Frobenius rings.
Details can be found in Lam~\cite[Th.~(16.14)]{Lam99}, Lamprecht~\cite{Lamp53}, Hirano~\cite[Th.~1]{Hi97}, and
Wood~\cite[Th.~3.10, Prop.~5.1]{Wo99}.
The most crucial aspect, namely that for finite rings the right-sided statements imply the left-sided ones,
has been proven by Honold~\cite[Th.~1 and Th.~2]{Hon01}.

\begin{theo}\label{T-Frob}
Let~$R$ be a finite ring. The following are equivalent.
\begin{alphalist}
\item $_R\,\soc(_R R)\cong\, _R(R/\rad(R))$.
\item $\soc(_R R)$ is a left principal ideal, i.e., $\soc(_R R)=Ra$ for some $a\in R$.
\item $_R\Rhat\cong\, _RR$.
\end{alphalist}
Each of the above is equivalent to the corresponding right-sided version.
The ring~$R$ is called Frobenius if any (hence all) of the above hold true.
In this case there exists a character~$\chi$ such that $\Rhat= R\!\cdot\!\chi$.
Any such character is called a generating character of~$R$ and also satisfies $\Rhat= \chi\!\cdot\!R$.
Any two generating characters~$\chi,\,\chi'$ differ by a unit, i.e., $\chi'=u\!\cdot\!\chi$ and
$\chi'=\chi\!\cdot\!u'$ for some $u,\,u'\in R^*$.
Furthermore, if~$R$ is Frobenius then the left and right socle coincide, and will be denoted by $\soc(R)$.
\end{theo}

The integer residue rings~$\Z_N$, finite fields, finite chain rings are Frobenius and so are
finite group rings as well as matrix ring over Frobenius rings.
The class of Frobenius rings is closed under taking direct products.
For details see Wood~\cite[Ex.~4.4]{Wo99} and Lam~\cite[Sec.~16.B]{Lam99}.

For the rest of this paper, let~$R$ be a finite Frobenius ring with group of units~$R^*$, and fix a generating
character~$\chi$ of~$R$.

The following definition of the homogeneous weight is taken from  Greferath and Schmidt~\cite{GrSch00}.
\begin{defi}\label{D-homogWt}
The \emph{(left) homogeneous weight on~$R$ with average value}~$\gamma\geq0$ is a map
$\omega:R\longrightarrow\R$ such that $\omega(0)=0$ and
\begin{romanlist}
\item $\omega(x)=\omega(y)$ for all $x,y\in R$ such that $Rx=Ry$,
\item $\sum_{y\in Rx}\omega(y)=\gamma|Rx|$ for all $x\in R\backslash\{0\}$; in other words, the average weight over each
      nonzero principal left ideal is~$\gamma$.
\end{romanlist}
The weight is called \emph{normalized} if $\gamma=1$.
\end{defi}

It has been shown in \cite[Th.~1.3]{GrSch00} that for any given~$\gamma\in\R_{\geq0}$ there exists a unique
homogeneous weight on~$R$ with average value~$\gamma$.
Of course, if~$R$ is a field of size~$q$, then the Hamming weight is the homogeneous weight with average value
$\frac{q-1}{q}$.
We next present the homogeneous weight for two specific cases.
Further examples can be found  in \cite[Sec.~3]{GL13homog}.

\begin{exa}\label{E-HomogWt}
\begin{arabiclist}
\item \cite[Ex.~2.8]{BGO07} If~$R$ is a local Frobenius ring with residue field $R/\rad(R)$ of order~$q$, then
      the normalized homogeneous weight is given by $\omega(0)=0$ and
      \[
         \omega(a)=\frac{q}{q-1}\text{ for } a\in \soc(R)\backslash\{0\}\ \text{ and }\
         \omega(a)=1\text{ otherwise.}
      \]
\item \cite[Ex.~2]{Byr11} and \cite[Ex.~4.5(b)]{GL13homog}
      If~$R=\Z_{pq}$, where~$p,\,q$ are distinct primes then the normalized homogeneous weight  is given by
      $\omega(0)=0$ and
      \[
         \omega(a)=
         \left\{\begin{array}{cl}\frac{p}{p-1},&\text{for }a\in p\Z_{pq}\backslash\{0\},\\[1ex]
                                \frac{q}{q-1},&\text{for }a\in q\Z_{pq}\backslash\{0\},\\[1ex]
                                \frac{pq-p-q}{(p-1)(q-1)},&\text{otherwise}.\end{array}\right.
      \]
\end{arabiclist}
\end{exa}

The following properties and explicit formulas will be crucial.
\begin{rem}\label{R-PropHomWt}
Let~$\omega$ be the normalized homogeneous weight on~$R$. Then
\begin{alphalist}
\item $\sum_{y\in I}\omega(y)=|I|$ for all nonzero left ideals~$I$ of~$R$, see ~\cite[Cor.~1.6]{GrSch00}.
\item Let~$\chi$ be a generating character of~$R$. Then by~\cite[p.~412, Th.~2]{Hon01}
      \[
           \omega(r)=1-\frac{1}{|R^*|}\sum_{u\in R^*}\chi(ru)=1-\frac{1}{|R^*|}\sum_{u\in R^*}\chi(ur)  \text{ for } r\in R.
      \]
\item The weight~$\omega$ is also right homogeneous, that is,~(i) and~(ii) of Definition~\ref{D-homogWt} are true for right ideals as well.
      This follows from~(b), see again \cite[Th.~2]{Hon01}.
\end{alphalist}
\end{rem}

In~\cite{GL13homog}, alternative expressions for the values of homogeneous weight and properties of the induced partition
were derived for products of local Frobenius rings.
In this paper we will study the homogeneous weight for arbitrary Frobenius rings.

As a first step we show that the normalized homogeneous weight is~$1$ for any element outside the socle.
To this end we make use of Theorem~\ref{T-Frob}(a).
Let
\begin{equation}\label{e-psi}
  \psi:\; _R\,\soc(R)\longrightarrow _R\!(R/\rad(R))
\end{equation}
be an isomorphism of left $R$-modules.
Recall that $\rad(R)$ is a two-sided ideal, hence  $R/\rad(R)$ is a ring.
Even more, since $R/\rad(R)$ is a finite semisimple ring~\cite[Th.~(4.14)]{Lam91}, the Wedderburn-Artin Theorem~\cite[Th.~(3.5)]{Lam91}
along with the fact that every finite division ring is a field
provides us with a ring isomorphism
\begin{equation}\label{e-RradR}
  R/\rad(R)\cong (\F_{q_1})^{m_1\times m_1}\times\ldots\times(\F_{q_t})^{m_t\times m_t}
\end{equation}
for suitable prime powers $q_1,\ldots,q_t$ and positive integers $m_1,\ldots,m_t$.
In particular, $R/\rad(R)$ is Frobenius.

Now we are ready to show that the homogeneous weight on~$R$ is constant outside the socle.
\begin{theo}\label{T-homogWtsoc}
Denote by~$\tilde{\omega}$ the normalized homogeneous weight on the ring $R/\rad(R)$.
Then the normalized homogeneous weight~$\omega$ on~$R$ is given by
\[
    \omega(x)=\left\{\begin{array}{cl} \tilde{\omega}(\psi(x)),&\text{ for }x\in\soc(R),\\[.5ex]
                                                 1,    &\text{ for }x\not\in\soc(R).\end{array}\right.
\]
In other words, $\omega(x)=1$ for all $x\in R\,\backslash\,\soc(R)$ and $\omega|_{\rm{soc}(R)}=\tilde{\omega}\circ\psi$.
\end{theo}
\begin{proof}
We show that~$\omega$ given above satisfies~(i) and~(ii) of Definition~\ref{D-homogWt}.
This is obvious for~(i).
For~(ii) we will show the constant average weight property for arbitrary left ideals (see Remark~\ref{R-PropHomWt}(a)),
and make use of the same property for the homogeneous weight~$\tilde{\omega}$ on~$R/\rad(R)$.

Thus, let~$I$ be a left ideal in~$R$.
If $I$ is  contained in~$\soc(R)$, then
\[
  \sum_{y\in I}\omega(y)=\sum_{y\in I}\tilde{\omega}(\psi(y))=\sum_{y\in\psi(I)}\tilde{\omega}(y)
  = |\psi(I)|=|I|.
\]
If~$I$ is not contained in $\soc(R)$, then $I\cap\text{soc}(R)$ is in $\soc(R)$, and from the previous part we obtain
\[
  \sum_{y\in I}\omega(y)=\sum_{y\in I\cap\text{soc}(R)}\omega(y)+\!\!\sum_{y\in I\backslash\text{soc}(R)}\!\!\!1\\
      =\big|I\cap\text{soc}(R)\big|+\big|I\,\backslash\,\soc(R)\big|=|I|.
  \qedhere
\]
\end{proof}

\begin{rem}\label{R-RightOmega}
Since any left homogeneous weight is also right homogeneous, we also obtain a right-sided version of the last theorem, that is
$\omega|_{\rm{soc}(R)}=\tilde{\omega}\circ\psi_r$, where $\psi_r: \soc(R)_R\longrightarrow (R/\rad(R))_R$ is an isomorphism of right
$R$-modules.
\end{rem}

In view of the previous result, it remains to compute the values of the homogeneous weight on the socle of~$R$.
Since these values are given by $\tilde{\omega}(\psi(x))$, this amounts to computing
the values for semisimple Frobenius rings.
This will be carried out in the next sections.

\section{The Homogeneous Weight on a Matrix Ring}
Let $\F=\F_q$ be any field of size~$q$, and let $R:=\F^{m\times m}$ be the matrix ring over~$\F$ of order~$m>1$.
It is well-known that~$R$ is simple.
We will present the values of the normalized homogeneous weight on~$R$ in terms of the rank.
We need to consider the principal left ideals in~$R$
(it is worth noting and not hard to see that~$R$ is even a left principal ideal ring, i.e., all left ideals are principal).
Clearly, for any matrices $A,\,B\in R$
\begin{equation}\label{e-LIdeals}
  RA=RB\Longleftrightarrow B=UA\text{ for some }U\in\GL_m(\F)\Longleftrightarrow \rowspace(A)=\rowspace(B).
\end{equation}
As a consequence, the left principal ideals in~$R$ are in one-to-one correspondence with the subspaces of~$\F^m$.
Recall that the number of $j$-dimensional subspaces in~$\F^m$ is given by the Gaussian coefficient
\begin{equation}\label{e-GaussCoeff}
   \Gaussian{m}{j}:=\Gaussian{m}{j}_q:=    \prod_{i=0}^{j-1}\frac{q^m-q^i}{q^j-q^i}\, \text{ for }j=0,\ldots,m.
\end{equation}
Hence~\eqref{e-LIdeals} shows that the number of left principal ideals in~$R$ is given by $\sum_{j=0}^m \Gaussian{m}{j}_q$.
The cardinality of any left principal ideal and the number of matrices of fixed rank within such an ideal are given as follows.

\begin{lemma}\label{L-RAcard}
Let $A\in R$ be a matrix of rank~$r$. Then $|RA|=q^{rm}$ and
$\big|\{B\in RA\mid \rank B=j\}\big|=s_j(m,r)$,
where
\[
    s_j(m,r)=\Gaussian{r}{j}\alpha_j(q^m)=\frac{\alpha_j(q^r)\alpha_j(q^m)}{\alpha_j(q^j)} \text{ and }
    \alpha_j(x):=\alpha_{q,j}(x)=\prod_{i=0}^{j-1}(x-q^i).
\]
As a consequence, $\sum_{j=0}^r s_j(m,r)=q^{rm}$.
\end{lemma}
The statement is also true for $j=0$ because $\Gaussian{r}{0}=1=\alpha_0(q^i)$ for all~$r$ and~$i$.
We will normally just write $\alpha_j$ instead of $\alpha_{q,j}$, and make the subscript~$q$ explicit
only  in the next sections when more than one field size is involved.

\begin{proof}
It is easy to see that we may assume without loss of generality that
\[
    A=\begin{pmatrix}I_r & 0\\0&0\end{pmatrix}, \ \text{ where $I_r$ is the $(r\times r)$-identity matrix.}
\]
Then the left ideal $RA$ consists of all matrices of the form $(M\,|\,0)$, where~$M$ is any matrix in~$\F^{m\times r}$.
This proves $|RA|=q^{rm}$.
Next,
$|\{B\in RA\mid \rank B=j\}|=|\{M\in\F^{m\times r}\mid \rank M=j\}|$.
This cardinality is given by $s_j(m,r)$, see \cite[Eq.~(2.9)]{Del78} or \cite[Prop.~3.1]{LaTh94}.
\end{proof}

Now we can determine the normalized homogeneous weight for matrices in~$R$ in terms of the rank.
We need the Cauchy binomial theorem, which states that the polynomial $\alpha_r$ from Lemma~\ref{L-RAcard} satisfies
the identity $\alpha_r(x)=\sum_{j=0}^r (-1)^j q^{\Binom{j}{2}}\Gaussian{r}{j}x^{r-j}$, see~\cite[Eq.~(13)]{Gab85} or \cite[p.~23]{LaTh94}.
Using $\alpha_r(1)=0$ this implies
\begin{equation}\label{e-Qid}
  \sum_{j=0}^r (-1)^j q^{\Binom{j}{2}}\Gaussian{r}{j}=0 \text{ for }r\geq1.
\end{equation}

For the rest of this section let~$\omega$ be the normalized homogeneous weight on $R=\F^{m\times m}$.

\begin{theo}\label{T-HomWtMat}
\[
   \omega(A)=\frac{(-1)^{r+1}q^{\Binom{r}{2}}}{\alpha_r(q^m)}+1, \text{ where } r=\rank(A).
\]
\end{theo}

\begin{proof}
We show that the map~$\omega$ as defined in the theorem satisfies the properties of a homogeneous weight.
First of all, it is clear that the zero matrix satisfies $\omega(0)=0$.
Secondly, if the matrices~$A,\,A'$ generate the same left ideal  they have the same rank, and thus
$\omega(A)=\omega(A')$.
It remains to show~(ii) of Definition~\ref{D-homogWt}.
Let $A\in R$ be of rank~$r\geq1$.
In view of Lemma~\ref{L-RAcard} we have to show that $\sum_{B\in RA}\omega(B)=q^{rm}$.
With the aid of  the same lemma we compute
\begin{align*}
  \sum_{B\in RA}\omega(B)&=
    \sum_{j=0}^r s_j(m,r)\Big(\frac{(-1)^{j+1}q^{\Binom{j}{2}}}{\alpha_j(q^m)}+1\Big)
     =\sum_{j=0}^r\Big((-1)^{j+1}q^{\Binom{j}{2}}\Gaussian{r}{j}+s_j(m,r)\Big)\\
    &=\sum_{j=0}^r s_j(m,r)+\sum_{j=0}^r(-1)^{j+1}q^{\Binom{j}{2}}\Gaussian{r}{j}
     =q^{rm},
\end{align*}
where the last step follows from~\eqref{e-Qid} and the fact that $\sum_{j=0}^r s_j(m,r)=q^{rm}$.
All of this shows that~$\omega$ is indeed the normalized homogeneous weight on~$R$.
\end{proof}

Theorem~\ref{T-HomWtMat} shows that matrices with the same rank have the same homogeneous weight.
This is also clear from the left- and right-invariance of the homogeneous weight.
As we show next, the specific formula for~$\omega(A)$ also implies the converse, that is, all matrices with the same homogeneous weight share the same rank.

\begin{cor}\label{C-HomRank}
Let $A,\,B\in R$. Then
$\omega(A)=\omega(B)$ if and only if $\rank(A)=\rank(B)$.
\end{cor}
\begin{proof}
It remains to prove the only-if-part.
Let $\omega(A)=\omega(B)$ and set $r:=\rank(A)$ and $s:=\rank(B)$.
Assume $r< s$.
Then
\[
    (-1)^{r+1}\frac{q^{\Binom{r}{2}}}{\alpha_r(q^m)}=(-1)^{s+1}\frac{q^{\Binom{s}{2}}}{\alpha_s(q^m)},
\]
and thus $r\equiv s~\mod 2$.
Then we obtain
\[
  q^{\Binom{s}{2}-\Binom{r}{2}}=\frac{\alpha_s(q^m)}{\alpha_r(q^m)}=\prod_{i=r}^{s-1}(q^m-q^i)=q^{\sum_{i=r}^{s-1}i}\prod_{j=r}^{s-1}(q^{m-j}-1)
  =q^{\frac{(s-1)s}{2}-\frac{(r-1)r}{2}}\prod_{j=r}^{s-1}(q^{m-j}-1).
\]
Hence we conclude that $\prod_{j=r}^{s-1}(q^{m-j}-1)=1$.
This implies, in particular, that the product involves at most one factor.
But this contradicts that $r\equiv s~\mod 2$.
\end{proof}

The previous result shows that the partition of~$\F^{m\times m}$ induced by the homogeneous weight coincides with the partition induced by the rank.
We will discuss these partitions in more detail in the next sections and therefore introduce now the necessary notation.

\medskip

A \emph{partition}~$\cP$ of~$R$ is a collection of non-empty and disjoint sets that cover~$R$.
The sets will be called \emph{blocks}, and the number of blocks in a partition~$\cP$ is denoted by
$|\cP|$.
If~$\cP$ consists of the~$M$ blocks~$P_i$, we will write $\cP=P_1\mid P_2\mid\ldots\mid P_M$.
Two partitions~$\cP$ and~$\cQ$ of~$R$ are \emph{equal} if their blocks coincide up to ordering.
The partition~$\cP$ is called \emph{finer} than the partition~$\cQ$, written $\cP\leq\cQ$, if every block of~$\cP$
is contained in a block of~$\cQ$.
We denote by $\widesim_{\cP}$ the equivalence relation induced by~$\cP$, thus $x\widesim_{\cP}y$ if and only
if~$x$ and~$y$ belong to the same block of~$\cP$.

\begin{defi}\label{D-Invpart}
A partition~$\cP=P_1\mid P_2\mid\ldots\mid P_M$ of~$R$ is called \emph{invariant} if each block of~$\cP$
is invariant under left or right multiplication by units, that is, $uP_m=P_m=P_mu$ for each $u\in R^*$.
\end{defi}

The following partition will be at the focus of the next sections.
\begin{defi}\label{D-homPart}
For any finite Frobenius ring~$R$, we denote by~$\cPhom$ the partition of~$R$ induced by the homogeneous
weight~$\omega$.
Thus, $x\widesim_{\cPhom} y\Longleftrightarrow \omega(x)=\omega(y)$ for all $x,y\in R$.
\end{defi}

The basic properties of the homogeneous weight and the previous result now translate as follows.
\begin{rem}\label{R-PhomFmm}
\begin{alphalist}
\item For any finite Frobenius ring~$R$, the partition $\cPhom$ is invariant.
\item If $R$ is the matrix ring $\F^{m\times m}$, then $\cPhom=\cPrank$, where~$\cPrank$ is the partition induced
         by the rank, i.e., $A\widesim_{\cPrank}B\Longleftrightarrow \rank(A)=\rank(B)$.
\end{alphalist}
\end{rem}

\section{The Homogeneous Weight on General Frobenius Rings}\label{SS-GenFrob}
In this section we first consider direct products of matrix rings over fields before moving on to general Frobenius rings.
For the first class of rings, the values of the homogeneous weight can easily be derived from the previous section
and a well known formula for direct product rings.
The general case then follows from Theorem~\ref{T-homogWtsoc}.

Throughout this section, any homogeneous weight is meant to be normalized.

%
%

\begin{theo}\label{T-homogWtProdmat}
Consider the ring~$R=R_1\times\ldots\times R_t$, where $R_i=(\F_{q_i})^{m_i\times m_i}$.
Then the homogeneous weight~$\omega$ on~$R$ is given by
\[
   \omega(A_1,\ldots,A_t)=
   1-\prod_{i=1}^t\frac{(-1)^{r_i}q_i^{\Binom{r_{i}}{2}}}{\alpha_{q_i,r_i}(q_i^{m_i})},\text{ where }r_i=\rank(A_i).
\]
\end{theo}

\begin{proof}
This follows from Theorem~\ref{T-HomWtMat} along with the product formula for the homogeneous weight on a direct
product of rings, as it can be found in \cite[Lem.~7]{BKS12} or \cite[Prop.~3.7]{GL13homog}.
\end{proof}

Now Theorem~\ref{T-homogWtsoc} leads to the following summary of our findings.
\begin{theo}\label{T-RFrob}
Let~$R$ be a finite Frobenius ring.
Then there exists a ring isomorphism
\begin{equation}\label{e-phi}
  \phi:\;R/\rad(R)\longrightarrow (\F_{q_1})^{m_1\times m_1}\times\ldots\times(\F_{q_t})^{m_t\times m_t},
\end{equation}
where $q_1,\ldots,q_t$ are suitable prime powers and $m_1,\ldots,m_t$ positive integers.
Denote by $\pi_i$ the projection of the product ring on the $i$-th component $(\F_{q_i})^{m_i\times m_i}$, and let
$\psi_i=\pi_i\circ\phi\circ\psi$, where~$\psi$ is the left $R$-module isomorphism in~\eqref{e-psi}.
Then the homogeneous weight~$\omega$ on~$R$ is given by
\[
    \omega(x)=\left\{\begin{array}{cl}
               {\DS 1-\prod_{i=1}^t\frac{(-1)^{r_i}q_i^{\Binom{r_{i}}{2}}}{\alpha_{q_i,r_i}(q_i^{m_i})}},&
                   \text{ if $x\in\soc(R)$ and $r_i=\rank(\psi_i(x))$ for } i=1,\ldots,t,\\[3.5ex]
                   1,&\text{ if }x\not\in\soc(R).\end{array}\right.
\]
As a consequence, $R\,\backslash\,\soc(R)$ is a block of the partition~$\cPhom$.
\end{theo}

The result allows us to characterize the rings that contain nonzero elements with zero homogeneous weight.

\begin{cor}\label{C-HomogWtZero}
Let~$R$ be a finite Frobenius ring. Then~$R$ contains a nonzero element with zero homogeneous weight if and
only if~$\F_2$ is a factor of multiplicity~$2$ in the Wedderburn-Artin decomposition of $R/\rad(R)$.
In other words, there exists some $x\in R\backslash\{0\}$ such that $\omega(x)=0$
if and only if $(q_i,\,m_i)=(2,1)$ for at least two values of~$i\in\{1,\ldots,t\}$ in~\eqref{e-phi}.
\end{cor}
\begin{proof}
Note that $\omega(x)=0$ iff
\begin{equation}\label{e-Product}
   \prod_{i=1}^t\frac{(-1)^{r_i}q_i^{\Binom{r_{i}}{2}}}{\alpha_{q_i,r_i}(q_i^{m_i})}
\end{equation}
is~$1$, where $r_i=\rank(\psi_i(x))$ as in the theorem.
\\
``$\Leftarrow$'' Assume $(q_1,m_1)=(q_2,m_2)=(2,1)$.
Let $x\in\soc(R)\backslash\{0\}$ such that $\psi_i(x)=0$ iff $i>2$.
Such~$x$ certainly exists due to the surjectivity of~$\phi$.
Since $\alpha_{2,1}(2)=1$, the product in~\eqref{e-Product} is~$1$, as desired.
\\
``$\Rightarrow$''
Assume there exists $x\neq0$ such that the product in~\eqref{e-Product} is~$1$. Fix~$i$ such that $m_i\geq2$.
Then if $r_i>0$, the expression $\alpha_{q_i,r_i}(q_i^{m_i})$ in the denominator
of~\eqref{e-Product} contains the factor $(q_i^{m_i}-1)$.
Since this factor is relatively prime to the numerator of~\eqref{e-Product}, this contradicts our assumption that the
product be~$1$, and we conclude $r_i=0$ whenever $m_i\geq2$.
Since~$x\neq0$, this implies that we must have at least one factor where $m_i=r_i=1$.
If $m_i=r_i=1$, then $\alpha_{q_i,r_i}(q_i^{m_i})=q_i-1$, and since~\eqref{e-Product} is~$1$, we conclude
$q_i=2$.
But then the factor $(-1)^{r_i}$ forces that there be at least two instances where $m_i=r_i=1$ and $q_i=2$.
This concludes the proof.
\end{proof}

\begin{exa}\label{R-locProd}
Suppose $R=R_1\times\ldots\times R_t$, where $R_i$ is a local Frobenius ring for all $i=1,\ldots,t$.
Then $R/\rad(R)\cong \prod_{i=1}^t R_i/\rad(R_i)$, where each component $R_i/\rad(R_i)$ is the residue field of~$R_i$.
So, $m_1=\ldots=m_t=1$ in the situation of Theorem~\ref{T-RFrob}.
Specify~\eqref{e-phi} to
\[
   \phi:\; R/\rad(R)\longrightarrow (\F_{q_1}\times\ldots\times\F_{q_1})\times\ldots\times(\F_{q_s}\times\ldots\times\F_{q_s}),
\]
where $q_1,\ldots,q_s$ are distinct, and the field $\F_{q_i}$ appears~$n_i$ times in the product.
Then we obtain for $x\in\soc(R)$ the formula $\omega(x)=1-\prod_{i=1}^s\Big(\frac{-1}{q_i-1}\Big)^{\wt(a_i)}$,  where
$(\phi\circ\psi)(x)=(a_1,\ldots,a_s)$  and $a_i=(a_{i,1},\ldots,a_{i,n_i})$
and where $\wt$ denotes the Hamming weight on each component $\F_{q_i}\times\ldots\times\F_{q_i}$.
This recovers~\cite[Th.~3.9]{GL13homog}.
\end{exa}

We close this section with two examples where we determine the homogeneous weight partition.
They will be revisited in the next section.

\begin{exa}\label{E-HomWt1}
Let $\F=\F_q$ and consider the ring $R=\F^{2\times2}\times \F^{2\times 2}$.
Theorem~\ref{T-homogWtProdmat} provides us with
\begin{equation}\label{e-wFq2}
    \omega(A_1,A_2)=1-\prod_{i=1}^2\frac{(-1)^{r_i}q^{\Binom{r_i}{2}}}{\alpha_{r_i}(q^2)},\text{ where }r_i=\rank(A_i).
\end{equation}
Let~$\cPhom$ be the induced homogeneous weight partition of~$R$.
Moreover, let $\cPrank=P_0\mmid P_1\mmid P_2$ be the rank partition of $\F^{2\times 2}$, thus
$P_i=\{A\in\F^{2\times 2}\mid\rank(A)=i\}$.
Define $\cQ:=(\cPrank)^2_{\text{sym}}$ to be the symmetrized
product partition of~$R$ induced by~$\cPrank$, that is,
its blocks are given by the pairs of matrices with the same ranks up to ordering
(see also \cite[Def.~3.2]{GL13pos} for general symmetrized product partitions).
To be precise, we index the blocks of~$\cQ$ by the multisets
$\{\!\{0,0\}\!\},\,\{\!\{0,1\}\!\},\,\{\!\{0,2\}\!\},\,\{\!\{1,1\}\!\},\,\{\!\{1,2\}\!\}$, $\{\!\{2,2\}\!\}$ so that
$Q_{\{\!\{i,j\}\!\}}$ consists of all matrix pairs with one matrix having rank~$i$ and
the other one having rank~$j$, regardless of the order.

We show now that $\cPhom=\cQ$ for $q>2$, while $\cPhom>\cQ$ for $q=2$.
It is clear from~\eqref{e-wFq2} that matrix pairs in the same block of~$\cQ$ have the same homogeneous weight.
In other words $\cQ\leq\cPhom$.
The values of the homogeneous weight on the blocks of~$\cQ$ are given by (in the above ordering of the multisets)
\[
   0,\, 1-\frac{-1}{q^2-1},\,1-\frac{1}{(q^2-1)(q-1)},\,1-\frac{1}{(q^2-1)^2},\,1-\frac{-1}{(q^2-1)^2(q-1)},\,1-\frac{1}{(q^2-1)^2(q-1)^2}.
\]
For $q>2$ these values are obviously distinct, and thus $\cPhom=\cQ$.
In other words, the homogeneous weight partition of~$R$ is given by the symmetrized rank partition.
For $q=2$, the matrix pairs in $Q_{\{\!\{1,1\}\!\}}\cup Q_{\{\!\{2,2\}\!\}}$ all have the same homogeneous weight, $\frac{8}{9}$,
and we obtain $\cPhom=Q_{\{\!\{0,0\}\!\}}\mmid Q_{\{\!\{0,1\}\!\}}\mmid Q_{\{\!\{0,2\}\!\}}\mmid Q_{\{\!\{1,2\}\!\}}\mmid Q_{\{\!\{1,1\}\!\}}\cup Q_{\{\!\{2,2\}\!\}}$.
Hence $\cQ<\cPhom$.
\end{exa}
\begin{exa}\label{E-HomWt3}
Let $R=\F^{2\times 2}\times\F$ where $\F=\F_q$.
Then the homogeneous weight of $(A,a)\in R$ is given by
\begin{equation}\label{e-RkWt}
    \omega(A,a)=1-\frac{(-1)^{r_1}q^{\Binom{r_1}{2}}(-1)^{r_2}}{\alpha_{r_1}(q^2)\alpha_{r_2}(q)},
    \text{ where }r_1=\rank(A)\text{ and }r_2=\wt(a)
\end{equation}
(and where wt is the Hamming weight, which equals the rank).
Thus pairs with the same rank and weight have the same homogeneous weight.
Let~$\cPrank$ again be the rank partition of~$\F^{2\times 2}$ and~$\cH$ be the Hamming weight partition of~$\F$.
Define~$\cQ$ as the product partition $\cPrank\times\cH$
(see also \cite[Def.~3.1]{GL13pos} for general product partitions).
Then $\cQ$ consists of the blocks
\[
  P_{(r_1,r_2)}:=\{(A,a)\mid \rank(A)=r_1,\,\wt(a)=r_2\} \text{ for all }(r_1,r_2)\in\{0,1,2\}\times\{0,1\}.
\]
As in previous example,~$\cQ\leq\cPhom$ due to\eqref{e-RkWt}.
But different from the previous example, we now observe  that the partition~$\cPhom$ is strictly coarser than~$\cQ$ for any~$q$.
Indeed, all pairs
in $P_{(1,1)}\cup P_{(2,0)}$ have homogeneous weight $1-\frac{1}{(q^2-1)(q-1)}$.
Checking all values of $\omega(A,a)$ in~\eqref{e-RkWt}, one arrives at the partition
\begin{equation}\label{e-PhomF22F}
   \cPhom=
   \left\{\begin{array}{ll}
      P_{(0,0)}\mmidbig P_{(0,1)}\mmidbig P_{(1,0)}\mmidbig P_{(2,1)}\mmidbig P_{(1,1)}\cup P_{(2,0)},&\text{ if }q>2,\\[1.8ex]
      P_{(0,0)}\mmidbig P_{(0,1)}\mmidbig P_{(1,0)}\cup P_{(2,1)} \mmidbig P_{(1,1)}\cup P_{(2,0)},&\text{ if }q=2.
   \end{array}\right.
\end{equation}
\end{exa}

\section{The Dual Partition of $\cPhom$}
In this section we will investigate the character-theoretic dual of the partition~$\cPhom$.
This dualization is at the heart of MacWilliams identities for the appropriate partition enumerators of codes
and their left or right dual.
The complex numbers $\sum_{a\in P_m}\chi(ab)$ and $\sum_{a\in P_m}\chi(ba)$, defined below, are
the left and right Krawtchouk coefficients and determine the MacWilliams transformation between the
partition enumerators.
See the introduction for further details and literature on this topic.
We follow the notation from~\cite{GL13pos}, where a general approach in the terminology of this paper has
been presented.

In this section we show that the left and right Krawtchouk coefficients coincide for the homogeneous weight partition,
and therefore the left and right dual partitions of $\cPhom$ coincide as well.
We will also provide an example showing that this is not true in general for invariant partitions.

As before, let~$R$ be a finite Frobenius ring and fix a generating character~$\chi$ of~$R$.

\begin{defi}[\mbox{\cite[Def.~2.1, Def.~4.1]{GL13pos}}]\label{D-DualPart}
Let~$\cP=P_1\mmid\ldots\mmid P_M$ be a partition of~$R$.
The \emph{left and right $\chi$-dual partition} of~$\cP$, denoted
by~$\wcP^{\scriptscriptstyle[\chi,l]}$ and ~$\wcP^{\scriptscriptstyle[\chi,r]}$, are defined by
the equivalence relations
\begin{equation}\label{e-simPhat2l}
  b\widesim_{\wcP^{\scriptscriptstyle[\chi,l]}} b' :\Longleftrightarrow
  \sum_{a\in P_m}\chi(ab)=\sum_{a\in P_m}\chi(ab') \text{ for all }m=1,\ldots,M
\end{equation}
and
\begin{equation}\label{e-simPhat2r}
  b\widesim_{\wcP^{\scriptscriptstyle[\chi,r]}} b' :\Longleftrightarrow
  \sum_{a\in P_m}\chi(ba)=\sum_{a\in P_m}\chi(b'a) \text{ for all }m=1,\ldots,M.
\end{equation}
$\cP$ is called \emph{$\chi$-self-dual} if $\cP=\wcPchil$ and \emph{reflexive} if
$\cP=\widehat{\phantom{\Big|}\hspace*{1.4em}}^{\,\scriptscriptstyle[\chi,r]}\hspace*{-3.5em}\widehat{\cP}^{^{\scriptscriptstyle[\chi,l]}}\quad $.
The sums $\sum_{a\in P_m}\chi(ab)$ (resp.\ $\sum_{a\in P_m}\chi(ba)$) for $b\in R,\,m=1,\ldots,M$,
are called the \emph{left $\chi$-Krawtchouk coefficients} (resp.\  \emph{right $\chi$-Krawtchouk coefficients}).
\end{defi}

It has been shown in \cite[Prop.~4.4]{BGL13} that $\chi$-self-duality does not depend on the sidedness of the dual partition, that is, $\cP=\wcPchil$ iff $\cP=\wcPchir$.
Moreover, reflexivity does not depend on the order of left and right duals taken and neither does it depend on the choice
of~$\chi$.
The dual partitions themselves do depend on the choice of~$\chi$ in general.
However, we have the following positive result.
\begin{rem}\label{R-InvPartDual}
Let~$\cP=P_1\mid\ldots\mid P_M$ be an invariant partition; see Definition~\ref{D-Invpart}.
Then the left (resp.\ right) Krawtchouk coefficients do not depend on the choice of~$\chi$, and thus neither do the left and right dual partition.
This can be seen as follows.
Let~$\chi,\,\chi'$ be two generating characters of~$R$.
Then $\chi'=u\!\cdot\!\chi=\chi\!\cdot\!v$ for some units $u,\,v\in R^*$ (see Theorem~\ref{T-Frob}).
Using~\eqref{e-bimodule} we obtain for any $b\in R$
\begin{equation}\label{e-chichi'}
  \sum_{a\in P_m}\chi'(ab)=\sum_{a\in P_m}\chi(vab)=\sum_{a'\in vP_m}\chi(a'b)=\sum_{a\in P_m}\chi(ab).
\end{equation}
\end{rem}

\medskip
\begin{cor}[see also \mbox{\cite[Rem.~3.3]{GL13homog}}]\label{C-DualPhomchi}
For the homogeneous weight partition~$\cPhom$ of~$R$, the left Krawt\-chouk coefficients and the left dual partition do not depend on the choice of the generating character~$\chi$.
The same is true for the right side.
We will therefore simply write $\wcPhoml$ and $\wcPhomr$ for these dual partitions.
\end{cor}

The goal of this section is to show that the left Krawtchouk coefficients coincide with the right Krawtchouk coefficients and thus
$\wcPhoml=\wcPhomr$.
Before addressing this question, we briefly sketch an example showing that this is not the case for general invariant partitions.
\begin{exa}\label{E-LRDual}
Consider the ring
\[
   R:=\left\{\begin{pmatrix}a&0&0&0\\0&a&b&0\\0&0&c&0\\d&0&0&c\end{pmatrix}\,\right|\,
      \left.\begin{matrix}\ \\ \ \\ \ \\ \ \end{matrix}\hspace*{-.6em} a,b,c,d\in\F_2\right\}
\]
(see also~\cite[Ex.~1.4(iii)]{Wo99}).
The ring is Frobenius, and a generating character is given by~$\chi$, defined by mapping the above matrix
to $(-1)^{a+b+c+d}$.
Consider the sets $P_0=\{0\},\;P_1=R^*,\;P_2=R^*A_1R^*\cup \{A_2\}$, and
$P_3=R^*B_1R^*\cup \{B_2,\,B_3\}$, where
\[
  A_1=\begin{pmatrix}0\!&\!0\!&\!0\!&\!0\\0\!&\!0\!&\!0\!&\!0\\0\!&\!0\!&\!1\!&\!0\\0\!&\!0\!&\!0\!&\!1\end{pmatrix}\!,
  A_2=\begin{pmatrix}0\!&\!0\!&\!0\!&\!0\\0\!&\!0\!&\!0\!&\!0\\0\!&\!0\!&\!0\!&\!0\\1\!&\!0\!&\!0\!&\!0\end{pmatrix}\!,
  B_1=\begin{pmatrix}1\!&\!0\!&\!0\!&\!0\\0\!&\!1\!&\!0\!&\!0\\0\!&\!0\!&\!0\!&\!0\\0\!&\!0\!&\!0\!&\!0\end{pmatrix}\!,
  B_2=\begin{pmatrix}0\!&\!0\!&\!0\!&\!0\\0\!&\!0\!&\!1\!&\!0\\0\!&\!0\!&\!0\!&\!0\\0\!&\!0\!&\!0\!&\!0\end{pmatrix}\!,
  B_3=\begin{pmatrix}0\!&\!0\!&\!0\!&\!0\\0\!&\!0\!&\!1\!&\!0\\0\!&\!0\!&\!0\!&\!0\\1\!&\!0\!&\!0\!&\!0\end{pmatrix}\!.
\]
It is easy to check that~$\cP=P_0\mid P_1\mid P_2\mid P_3$ is an invariant partition of~$R$.
However, $\widehat{\,\cP\,}^{\scriptscriptstyle l}\neq \widehat{\,\cP\,}^{\scriptscriptstyle r}$, which takes
a bit more effort to verify.
It should be noted that the ring~$R$ is not semisimple.
In Corollary~\ref{C-PhomBothDuals} below we will see that for semisimple rings the left and right dual partitions of an
invariant  partition always coincide.
\end{exa}


We now turn to the homogeneous weight partition.
Wen dealing with reflexivity we will make use of the following reflexivity criterion from~\cite[Th.~2.4]{GL13pos}
(see also~\cite[Fact~V.2]{Hon10} and \cite[Th.~10.1]{God93}).
\begin{prop}\label{P-Refl}
For any partition~$\cP$ of~$R$ we have $|\cP|\leq|\wcPchil|$ with equality if and only if~$\cP$ is reflexive.
The same is true for the right dual partition.
\end{prop}

The following concept will be helpful.
\begin{defi}\label{D-CommChar}
Let~$R$ be a finite Frobenius ring.
A character~$\chi$ of~$R$ is called \emph{symmetric} if $\chi(ab)=\chi(ba)$ for all $a,\,b\in R$.
\end{defi}

For semisimple Frobenius rings symmetric generating characters exist.
\begin{theo}\label{T-CommChar}
Let~$R$ be a semisimple finite Frobenius ring.
Then there exists a symmetric generating character of~$R$.
\end{theo}
\begin{proof}
Note first that if $\phi:R\longrightarrow S$ is a ring isomorphism and~$\chi$ is a symmetric generating character of~$S$, then
$\chi\circ\phi$ is a symmetric generating character of~$R$.
Symmetry is clear and the generating property follows from the fact that a character is generating if
and only if its kernel does not contain any nonzero left or right ideals~\cite[Cor.~3.6]{ClGo92}.
Thus we may use the Wedderburn-Artin Theorem \cite[Th.~(3.5)]{Lam91} and assume that~$R$ is a
product of matrix rings over finite fields.
\\
1) Let us first assume that $R=\F^{m\times m}$ for some field~$\F$.
Let~$\tilde{\chi}$ be a generating character of~$\F$, and denote by $\tr(A)$ the trace of the matrix~$A\in\F^{m\times m}$.
For $A\in R$ define $\chi(A):=\tilde{\chi}(\tr(A))$.
Then it is clear that~$\chi$ is a character of~$R$ which is symmetric due to the commutativity of the trace.
Thus it remains to show that~$\chi$ is a generating character.
In order to do so, suppose $\chi\!\cdot\!M$ is the trivial character on~$R$ for some matrix~$M\in R$.
We have to show that $M=0$; see Theorem~\ref{T-Frob}.
By assumption and~\eqref{e-bimodule} we have $\tilde{\chi}(\tr(MA))=1$ for all $A\in R$.
Using for $A$ the matrices $\alpha E_{j,i}$, where $\alpha\in\F$ and
$E_{j,i}$ is the matrix with entry~$1$ at position $(j,i)$ and zeros elsewhere, we arrive at
$\tilde{\chi}(M_{ij}\alpha)=1$ for all $\alpha\in\F$.
Hence $\tilde{\chi}\!\cdot\!M_{ij}$ is the trivial character on~$\F$ (see~\eqref{e-bimodule}), and thus $M_{ij}=0$
because $\tilde{\chi}$ is
generating.
Since this is true for all entries $M_{ij}$, we conclude $M=0$ and thus~$\chi$ is generating.
\\
2) Let now $R=(\F_{q_1})^{m_1\times m_1}\times\ldots\times(\F_{q_t})^{m_t\times m_t}$.
For each~$j$ let~$\chi_j$ be a generating character on $(\F_{q_j})^{m_j\times m_j}$.
It is easy to see that~$\chi$, defined as
$\chi(A_1,\ldots,A_t):=\prod_{j=1}^t \chi_j(A_j)$ for $(A_1,\ldots,A_t)\in R$,
is a generating character of~$R$.
Using~$\chi_j$ as in part~1), we conclude
$\chi(AB\big)=\chi(BA)$ for all $A=(A_1,\ldots,A_t)$ and $B=(B_1,\ldots,B_t)\in R$.
\end{proof}

\begin{cor}\label{C-PhomBothDuals}
Let~$R$ be a semisimple finite Frobenius ring, and let~$\cP$ be an invariant partition of~$R$.
Then the left and right Krawtchouk coefficients of~$\cP$ coincide, and thus
$\widehat{\,\cP\,}^{\scriptscriptstyle l}=\widehat{\,\cP\,}^{\scriptscriptstyle r}$.
In particular, $\wcPhoml=\wcPhomr$.
\end{cor}
\begin{proof}
This follows from Theorem~\ref{T-CommChar} and Definition~\ref{D-DualPart}.
\end{proof}

The methods above also provide us with the self-duality of~$\cPhom$ on matrix rings.
\begin{theo}\label{T-PhomMat}
The homogeneous weight partition on $\F^{m\times m}$ is self-dual, that is, $\cPhom=\wcPhoml$.
\end{theo}

\begin{proof}
By virtue of Theorem~\ref{T-CommChar} there exists a symmetric generating character, say~$\chi$.
From Corollary~\ref{C-HomRank} we know that $\cPhom=P_0\mmid\ldots\mmid P_m$, where $P_j=\{A\in R\mid \rank(A)=j\}$.
Let now $B,B'\in P_r$ for some $r$; hence $B\widesim_{\cPhom}B'$. Then $B'=UBV$ for some $U,\,V\in GL_m(\F)$ and thus
\[
  \sum_{A\in P_j}\chi(AB')=\sum_{A\in P_j}\chi(AUBV)=\sum_{A\in P_j}\chi(VAUB)=\sum_{A\in P_j}\chi(AB),
\]
where the second step follows from the symmetry of~$\chi$ and the last one from the invariance of~$\cPhom$.
This shows that $B\widesim_{\wcPhoml}B'$ and thus the partition~$\cPhom$ is finer than or equal to $\wcPhoml$.
This means $|\cPhom|\geq |\wcPhoml|$, and now Proposition~\ref{P-Refl} establishes $\cPhom=\wcPhoml$.
\end{proof}

We wish to add that self-duality  of the rank partition on $\F^{m\times m}$  has been shown earlier in the context of
abelian association schemes in \cite[Ex.~4.66]{Cam98} with a suitable identification between the primal and dual
scheme involved.

\medskip
The examples from Section~\ref{SS-GenFrob} illustrate that the homogeneous weight partition on a semisimple Frobenius ring
is in general not self-dual and not even reflexive.
\begin{exa}\label{E-MatRing}
Consider the ring $R=\F^{2\times2}\times\F^{2\times2}$ from Example~\ref{E-HomWt1}.
Let~$\cPrank=P_0\mmid P_1\mmid P_2$ be the rank partition of $\F^{2\times 2}$, which, as we know, coincides
with the homogeneous weight partition on~$\F^{2\times 2}$.
We have seen already that if $q=|\F|>2$, then the homogeneous weight partition $\cPhom$ of~$R$ coincides with the
symmetrized product partition of~$\cPrank$, i.e., $\cPhom=(\cPrank)_{\rm{sym}}^2$, while for $q=2$
$\cPhom$ is strictly coarser than $(\cPrank)_{\rm{sym}}^2$.
For the dual partition we have the following results.
\begin{alphalist}
\item If $q>2$ then $\cPhom$ is self-dual.
      This follows from the fact that the partition~$\cPrank$ is self-dual (see Theorem~\ref{T-PhomMat}), and thus
      the same is true for $(\cPrank)_{\rm{sym}}^2=\cPhom$, see \cite[Th.~3.3(b)]{GL13pos}.
\item If $q=2$, one can verify (using a computer algebra system), that the dual of~$\cPhom$ coincides
      with $(\cPrank)_{\rm{sym}}^2$.
      Hence ~$\cPhom$ is not reflexive due to Proposition~\ref{P-Refl}.
\end{alphalist}
\end{exa}

\begin{exa}\label{E-MatRing2}
Consider the ring $R=\F^{2\times2}\times\F$ from Example~\ref{E-HomWt3}.
The homogeneous partition has been determined in~\eqref{e-PhomF22F}.
Recall the partition
$\cQ=P_{(0,0)}\mmid P_{(0,1)}\mmid P_{(1,0)}\mmid P_{(1,1)}\mmid P_{(2,0)}\mmid P_{(2,1)}$, which is the product
of the rank partitions on~$\F^{2\times2}$ and~$\F$.
We saw already that $\cQ<\cPhom$ for each~$q$.
As for the dual partition, we have the following results.
\begin{alphalist}
\item Let $q>2$. Then $\wcPhom=\cQ$.
      Thus, $\wcPhom<\cPhom$ and~$\cPhom$ is not reflexive.
      In order to establish the identity $\wcPhom=\cQ$ one first observes that the inequality $\cQ<\cPhom$ along with
      $\cQ=\widehat{\cQ}$ implies $\cQ\leq\wcPhom$, see also \cite[Rem.~2.2(c)]{GL13pos}.
      For the converse, that is $\wcPhom\leq\cQ$, one has to show that for any two elements $(A,a),\,(A',a')$ in different blocks of~$\cQ$ there
      exists a block~$P$ of~$\cPhom$ such that
      \[
        \sum_{(B,b)\in P}\chi\big((A,a)\cdot(B,b)\big)\neq \sum_{(B,b)\in P}\chi\big((A',a')\cdot(B,b)\big),
      \]
      and where~$\chi$ is a generating character of~$R$.
      This can be verified by making use of the Krawtchouk coefficients of the rank partition of~$\F^{m\times m}$ as they have been derived by
      Delsarte~\cite[Eq.~(A10)]{Del78}.
      They tells us that for any $A\in\F^{m\times m}$ with rank~$i$
      \[
        \sum_{\text{rk}(B)=k}\tilde{\chi}(\tr(BA))=\sum_{j=0}^m(-1)^{k-j}q^{jm+\Binom{k-j}{2}}\Gaussian{m-j}{m-k}\Gaussian{m-i}{j},
      \]
      where~$\tilde{\chi}$ is a generating character of~$\F$.
\item Let $q=2$. Then $\wcPhom=P_{(0,0)}\mmid P_{(0,1)}\cup P_{(1,1)}\mmid P_{(1,0)}\cup P_{(2,0)}\mmid P_{(2,1)}$.
      This is derived in a similar manner as in~(a).
      Hence we see that $|\cPhom|=4=|\wcPhom|$ and therefore~$\cPhom$ is reflexive (but not self-dual).
\end{alphalist}
\end{exa}

The last example is particularly interesting when compared to the situation for finite
Frobenius rings that are direct products of fields, say
$R=(\F_{q_1}\times\ldots\times\F_{q_1})\times\ldots\times(\F_{q_t}\times\ldots\times\F_{q_t})$ for distinct~$q_i$.
In this case it has been shown in \cite[Th.~4.4 and Th.~4.7]{GL13homog} that the dual of the homogeneous weight partition of~$R$
is given by the product of the Hamming partitions on the components
$\F_{q_i}\times\ldots\times\F_{q_i}$.
Moreover,~$\cPhom$ is reflexive if and only if it is self-dual.
This contrasts Example~\ref{E-MatRing2}(b) where~$\wcPhom$ is not a product
partition, yet $\cPhom$ is reflexive.

\medskip
The above examples show that it does not seem easy to characterize the rings for which the homogeneous weight partition is reflexive.
We leave this for future research.

\medskip
We close the paper with proving that for any finite Frobenius ring, the left and right Krawtchouk coefficients of~$\cPhom$
coincide, and thus  $\wcPhoml=\wcPhomr$.
Recall that general Frobenius rings do not necessarily have a symmetric character, and the left dual of a partition does
in general not agree with the right dual one -- even if the partition is invariant;  see Example~\ref{E-LRDual}.

The crucial fact that makes the homogeneous weight partition stand out stems from Theorem~\ref{T-homogWtsoc} and Remark~\ref{R-RightOmega}.
Those results not only allow us to reduce the partition to the (semisimple) ring $R/\rad(R)$, but also ensure that the reduction
does not depend of the sidedness of the module isomorphism between $\soc(R)$ and $R/\rad(R)$.

\begin{theo}\label{T-Phomdual}
Let~$R$ be any finite Frobenius ring, and let $\cPhom=P_0\mid P_1\mid\ldots\mid P_M$ be the homogeneous weight partition.
Then for each generating character~$\chi$ of~$R$ we have
\begin{equation}\label{e-chiabba}
  \sum_{a\in P_m}\chi(ab)=\sum_{a\in P_m}\chi(ba)\text{ for all }b\in R\text{ and }m=0,\ldots,M.
\end{equation}
As a consequence, $\wcPhoml=\wcPhomr$.
\end{theo}
\begin{proof}
It  suffices to show~\eqref{e-chiabba}.
Let $S:=R/\rad(R)$ and $\pi:\,R\longrightarrow S$ be the canonical projection.
Since~$R$ is Frobenius, we have right and left $R$-module isomorphisms
\[
  \psi_r:\;\soc(R)_R\longrightarrow S_R,\quad \psi_l:\;_R\,\soc(R)\longrightarrow _R\!\!S.
\]
Recall from Theorem~\ref{T-RFrob} that $R\,\backslash\,\soc(R)$ is a block of~$\cPhom$.
Thus we may assume that the block $P_0$ of $\cPhom$ is $P_0=R\,\backslash\,\soc(R)$.
Then $\cP':=P_1\mmid\ldots\mmid P_M$ is a partition of $\soc(R)$.
Denote by~$\omega,\,\tilde{\omega}$ the normalized homogeneous weights on~$R$ and~$S$, respectively.
Theorem~\ref{T-homogWtsoc} and Remark~\ref{R-RightOmega} show that
$\tilde{\omega}(\psi_l(x))=\omega(x)=\tilde{\omega}(\psi_r(x))$ for all $x\in\soc(R)$.
Thus
\begin{equation}\label{e-leftrightPm}
    \psi_r(P_m)=\psi_l(P_m)\text{ for all }m=1,\ldots,M.
\end{equation}
Moreover, the partition $\cQ:=Q_1\mid\ldots\mid Q_M$, where $Q_m:=\psi_r(P_m)$, is the homogeneous weight partition
of the ring~$S$.

Let~$\chi$ be a generating character of~$R$.
Then $\tilde{\chi}_l:=\chi\circ\psi_l^{-1}$ and $\tilde{\chi}_r:=\chi\circ\psi_r^{-1}$
are both generating characters of~$S$.
This follows from the fact that a character is generating if and only if its kernel does not contain any nonzero left or right
ideals~\cite[Cor.~3.6]{ClGo92}.

Let now $b\in R$.
If we can show that
\begin{equation}\label{e-chiab}
    \sum_{a\in P_m}\chi(ab)=\sum_{a\in P_m}\chi(ba)\text{ for all }m=1,\ldots,M,
\end{equation}
then the remaining identity for $m=0$, and hence~\eqref{e-chiabba}, follows from the orthogonality relations
$\sum_{a\in R}\chi(ab)=|R|\delta_{b,0}=\sum_{a\in R}\chi(ba)$ (with the Kronecker symbol~$\delta$).
Before we start the computation we collect some facts.
\begin{liste}
\item By the very definition of generating characters (Theorem~\ref{T-Frob}) together with Theorem~\ref{T-CommChar} there exist units
      $u,v\in S^*$
      such that $v\!\cdot\!\tilde{\chi}_l=\tilde{\chi}_r\!\cdot\!u$, and such that this is a symmetric character of~$S$.
\item $\psi_r(ab)=\psi_r(a)\pi(b)$ and $\psi_l(ba)=\pi(b)\psi_l(a)$ for all $b\in R$ and $a\in\soc(R)$.
      This is clear from the simultaneous ring and bimodule structure of~$S$.
\end{liste}
Making use of these properties along with~\eqref{e-bimodule} and the invariance of the Krawtchouk sums from the chosen character,
see~\eqref{e-chichi'}, we compute for any $m\in\{1,\ldots,M\}$
\begin{align*}
  \sum_{a\in P_m}\chi(ab)&=\sum_{a\in P_m}\tilde{\chi}_r\big(\psi_r(ab)\big)=\sum_{a\in P_m}\tilde{\chi}_r\big(\psi_r(a)\pi(b)\big)
        =\sum_{\tilde{a}\in\psi_r(P_m)}\!\!\!\tilde{\chi}_r\big(\tilde{a}\pi(b)\big)
        \\[.6ex]
      &=\sum_{\tilde{a}\in\psi_r(P_m)}\!\!\!(\tilde{\chi}_r\!\cdot\!u)\big(\tilde{a}\pi(b)\big)
       =\sum_{\tilde{a}\in\psi_l(P_m)}\!\!\!(v\!\cdot\!\tilde{\chi}_l)\big(\tilde{a}\pi(b)\big)
       =\sum_{\tilde{a}\in\psi_l(P_m)}\!\!\!(v\!\cdot\!\tilde{\chi}_l)\big(\pi(b)\tilde{a}\big)\\[.6ex]
      &
       =\sum_{\tilde{a}\in\psi_l(P_m)}\!\!\!\tilde{\chi}_l\big(\pi(b)\tilde{a}\big)
       =\sum_{a\in P_m}\tilde{\chi}_l\big(\psi_l(ba)\big)
       =\sum_{a\in P_m}\chi(ba).
\end{align*}
This concludes the proof.
\end{proof}

As already mentioned earlier, we leave it to future research to characterize the rings for which the homogeneous
weight partition~$\cPhom$ is reflexive or even self-dual.

\bibliographystyle{abbrv}

\begin{thebibliography}{10}

\bibitem{BGL13}
A.~Barra and H.~Gluesing-Luerssen.
\newblock Mac{W}illiams extension theorems and the local-global property for
  codes over {F}robenius rings.
\newblock Preprint 2013. Submitted. arXiv:1307.7159v1 [cs.IT].

\bibitem{Byr11}
E.~Byrne.
\newblock On the weight distribution of codes over finite rings.
\newblock {\em Adv. Math. Commun.}, 5:395--406, 2011.

\bibitem{BGO07}
E.~Byrne, M.~Greferath, and M.~E. O'Sullivan.
\newblock The linear programming bound for codes over finite {F}robenius rings.
\newblock {\em Des. Codes Cryptogr.}, 42:289--301, 2007.

\bibitem{BKS12}
E.~Byrne, M.~Kiermaier, and A.~Sneyd.
\newblock Properties of codes with two homogeneous weights.
\newblock {\em Finite Fields Appl.}, 18:711--727, 2012.

\bibitem{Cam98}
P.~Camion.
\newblock Codes and association schemes.
\newblock In V.~S. Pless and W.~C. Huffman, editors, {\em Handbook of Coding
  Theory, Vol.~II}, pages 1441--1566. Elsevier, Amsterdam, 1998.

\bibitem{ClGo92}
H.~L. Claasen and R.~W. Goldbach.
\newblock A field-like property of finite rings.
\newblock {\em Indag.\ Math.}, 3:11--26, 1992.

\bibitem{CoHe97}
I.~Constantinescu and W.~Heise.
\newblock A metric for codes over residue class rings.
\newblock {\em Problems Inform.\ Transmission}, 33:208--213, 1997.

\bibitem{Del73}
P.~Delsarte.
\newblock An algebraic approach to the association schemes of coding theory.
\newblock {\em Philips Res.\ Repts.\ Suppl.}, 10, 1973.

\bibitem{Del78}
P.~Delsarte.
\newblock Bilinear forms over a finite field, with applications to coding
  theory.
\newblock {\em J.\ Combin. Theory Ser.\ A}, 25:226--241, 1978.

\bibitem{Gab85}
E.~M. Gabidulin.
\newblock Theory of codes with maximal rank distance.
\newblock {\em Probl. Inf. Transm.}, 21:1--12, 1985.

\bibitem{GL13pos}
H.~Gluesing-Luerssen.
\newblock Fourier-reflexive partitions and {M}ac{W}illiams identities for
  additive codes.
\newblock Preprint 2013. To appear in \emph{Designs, Codes and Cryptography}.
  arXiv: 1304.1207v1 [cs.IT].

\bibitem{GL13homog}
H.~Gluesing-Luerssen.
\newblock Partitions of {F}robenius rings induced by the homogeneous weight.
\newblock Preprint 2013. To appear in \emph{Adv. Math. Commun.} arXiv:
  1304.6589v1 [cs.IT].

\bibitem{God93}
C.~D. Godsil.
\newblock {\em Algebraic Combinatorics}.
\newblock Chapman and Hall, New York, 1993.

\bibitem{GrSch00}
M.~Greferath and S.~E. Schmidt.
\newblock Finite ring combinatorics and {M}ac{W}illiams' {E}quivalence
  {T}heorem.
\newblock {\em J.~Combin.\ Theory Ser.~A}, 92:17--28, 2000.

\bibitem{Hi97}
Y.~Hirano.
\newblock On admissible rings.
\newblock {\em Indag.\ Math.}, 8:55--59, 1997.

\bibitem{Hon01}
T.~Honold.
\newblock Characterization of finite {F}robenius rings.
\newblock {\em Arch. Math.}, 76:406--415, 2001.

\bibitem{Hon10}
T.~Honold.
\newblock Two-{I}ntersection sets in projective {H}jelmslev spaces.
\newblock In {\em Proceedings of the 19th International Symposium on the
  Mathematical Theory of Networks and Systems}, pages 1807--1813, Budapest,
  Hungary, 2010.

\bibitem{HoLa01}
T.~Honold and I.~Landjev.
\newblock Mac{W}illiams identities for linear codes over finite {F}robenius
  rings.
\newblock In D.~Jungnickel and H.~Niederreiter, editors, {\em Proceedings of
  The Fifth International Conference on Finite Fields and Applications Fq5
  (Augsburg, 1999)}, pages 276--292. Springer, Berlin Heidelberg New York,
  2001.

\bibitem{LaTh94}
D.~Laksov and A.~Thorup.
\newblock Counting matrices with coordinates in finite fields and of fixed
  rank.
\newblock {\em Math. Scand.}, 74:19--33, 1994.

\bibitem{Lam91}
T.~Y. Lam.
\newblock {\em A First Course in Noncommutative Rings}.
\newblock Graduate Text in Mathematics, Vol.~131. Springer, 1991.

\bibitem{Lam99}
T.~Y. Lam.
\newblock {\em Lectures on Modules and Rings}.
\newblock Graduate Text in Mathematics, Vol.~189. Springer, 1999.

\bibitem{Lamp53}
E.~Lamprecht.
\newblock {\"U}ber {I}-regul{\"a}re {R}inge, regul{\"a}re {I}deale and
  {E}rkl{\"a}rungsmoduln.
\newblock {\em I. Math. Nachr.}, 10:353--382, 1953.

\bibitem{Wo99}
J.~A. Wood.
\newblock Duality for modules over finite rings and applications to coding
  theory.
\newblock {\em Americ.\ J.\ of Math.}, 121:555--575, 1999.

\bibitem{ZiEr96}
V.~A. Zinoviev and T.~Ericson.
\newblock On {F}ourier invariant partitions of finite abelian groups and the
  {M}ac{W}illiams identity for group codes.
\newblock {\em Problems Inform. Transmission}, 32:117--122, 1996.

\bibitem{ZiEr09}
V.~A. Zinoviev and T.~Ericson.
\newblock {F}ourier invariant pairs of partitions of finite abelian groups and
  association schemes.
\newblock {\em Problems Inform. Transmission}, 45:221--231, 2009.

\end{thebibliography}


\end{document}